\documentclass[orivec]{llncs}

\usepackage[utf8]{inputenc}
\usepackage{graphicx}
\usepackage{subcaption}
\usepackage{amsmath,amssymb}
\usepackage{booktabs}
\usepackage{array}
\usepackage{enumitem}
\usepackage{cite}
\usepackage[pdfpagelabels,colorlinks,allcolors=blue]{hyperref}

\newcommand{\df}[1]{{\it #1}}
\newcommand{\dfb}[1]{{\it \bf #1}}
\newcommand{\Oh}{{\ensuremath{\mathcal{O}}}}
\newcommand{\Sh}{{\ensuremath{\mathcal{S}}}}
\newcommand{\Qh}{{\ensuremath{\mathcal{Q}}}}
\newcommand{\eps}{\ensuremath{\varepsilon}}

\newcommand{\Gc}{{\ensuremath{\overline{G}}}}

\let\doendproof\endproof
\renewcommand\endproof{~\hfill\qed\doendproof}

\newtheorem{conj}{Conjecture}
\newtheorem{cs}{Case}

\newcounter{dummycount}

\newcommand{\wormholeThm}[1]{
    \newcounter{#1}
    \setcounter{#1}{\value{theorem}}}

\newenvironment{backInTimeThm}[1]{
    \setcounter{dummycount}{\value{theorem}}
    \setcounter{theorem}{\value{#1}}}
{\setcounter{theorem}{\value{dummycount}}}

\begin{document}
    
\title{Mixed Linear Layouts of Planar Graphs}

\author{
    Sergey Pupyrev
}

\institute{\email{spupyrev@gmail.com}}

\date{}
\maketitle

\begin{abstract}
    A $k$-stack (respectively, $k$-queue) layout of a graph consists of a total order of the vertices, and a partition
    of the edges into $k$ sets of non-crossing (non-nested) edges with respect to the vertex ordering.
    In 1992, Heath and Rosenberg conjectured that every planar graph admits a mixed $1$-stack $1$-queue layout in which
    every edge is assigned to a stack or to a queue that use a common vertex ordering.
    
    We disprove this conjecture by providing a planar graph that does not have such a mixed layout.
    In addition, we study mixed layouts of graph subdivisions, and show that every planar graph has a mixed
    subdivision with one division vertex per edge.
\end{abstract}

\section{Introduction}

A \df{stack layout} of a graph consists of a linear order on the vertices and an assignment of the
edges to \df{stacks}, such that no two edges in a single stack cross. A ``dual'' concept is
a \df{queue layout}, which is defined similarly, except that no two edges in a single queue may
be nested. The minimum number of stacks (queues) needed in a stack layout (queue layout) of a graph
is called its \df{stack number} (\df{queue number}). Stack and queue layouts were respectively introduced by 
Ollmann~\cite{Oll73} and Heath et al.~\cite{HR92,HLR92}. These are ubiquitous structures with a variety of applications, 
including complexity theory, VLSI design, bioinformatics, parallel process
scheduling, matrix computations, permutation sorting, and graph drawing; see~\cite{DW04} for more details.

Stack and queue layouts have been extensively studied for planar graphs. The stack number of a graph, also known as
\df{book thickness}, is one if and only if the graph is outerplanar~\cite{BK79}.
The stack number of a graph $G$ is at most two if and only if $G$ is subhamiltonian, that is,
a subgraph of a planar graph that has a Hamiltonian cycle~\cite{BK79}. More generally, all planar graphs have stack
number at most four~\cite{Yan89}. Similarly, every graph admitting a $1$-queue layout is planar with an
``arched leveled-planar'' embedding~\cite{HLR92}. Many subclasses of planar graphs have bounded 
queue number: Every tree has queue number one~\cite{HR92}, outerplanar graphs have queue number at most two~\cite{HLR92};
series-parallel graphs have queue number at most three~\cite{RM95}, and planar 3-trees have queue number at most
seven~\cite{Wie17}. It is, however, an open question whether every planar graph have a constant
queue number; Dujmovi{\'c} shows that planar graphs have queue number $\Oh(\log n)$~\cite{Duj15}, improving 
an earlier result of $\Oh(\log^4 n)$ by Di Battista et al.~\cite{BFP13}.

Stack and queue layouts are generalized through the notion of a \df{mixed} layout, in which
every edge is assigned to a stack or to a queue that is defined with respect to a common
vertex ordering~\cite{HR92}. Such a layout is called an \df{$s$-stack $q$-queue} layout, if it utilizes $s$ stacks and $q$
queues. One reason for studying mixed stack and queue layouts is that they model the dequeue data structure, as
a dequeue may be simulated by two stacks and one queue~\cite{DW05,Auer14}.
Here we study mixed layouts of planar graphs.

In their seminal paper~\cite{HR92}, Heath and Rosenberg make the following conjecture, which
has hitherto been unresolved.

\begin{conj}[Heath and Rosenberg~\cite{HR92}]
    \label{conj:HR}
    Every planar graph admits a \break mixed $1$-stack $1$-queue layout.
\end{conj}

In this paper we disprove the conjecture by providing a planar graph that does not have 
a $1$-stack $1$-queue layout.

\wormholeThm{thm-ce}
\begin{theorem}
    \label{thm:ce}
    There exists a planar graph that does not admit a mixed $1$-stack $1$-queue layout.
\end{theorem}

We found, however, that mixed layouts are rather ``powerful''. Our experimental evaluation indicates that
\df{all} planar graphs with $|V| \le 18$ vertices admit a $1$-stack $1$-queue layout. This is in contrast
with pure stack and queue layouts: There exists a $11$-vertex planar graph that requires three stacks, and
there exists $14$-vertex planar graphs that requires three queues. Thus a reasonable question is what subclasses of 
planar graphs admit a $1$-stack $1$-queue layout. Dujmovi{\'c} and Wood~\cite{DW05} consider graph
subdivisions, that is, graphs created by replacing every edge of a graph by a path; they show that every 
planar graph has a $1$-stack $1$-queue 
subdivision with four division vertices per edge. We strengthen this result by showing that one division vertex
per edge is sufficient.

\wormholeThm{thm-11}
\begin{theorem}
    \label{thm:11}
    Every planar graph admits a mixed $1$-stack $1$-queue subdivision with one division vertex per edge.
\end{theorem}

\paragraph{Proof Ideas and Organization.}

Our construction of the counterexample for Conjecture~\ref{conj:HR} (presented in Section~\ref{sect:ce}) is based
on a sequence of gadgets~---~planar graphs that do not admit a mixed layout under certain conditions. We start
with a relatively simple gadget whose linear layouts can be analyzed exhaustively; this gadget does not
admit a mixed layout under fairly strong conditions. Several small gadgets are combined into a bigger one, that does not
have a mixed layout under weaker conditions; these bigger gadgets are combined together to produce the final counterexample.
We believe that such an approach is general and can be used for creating other lower bounds in the context of linear layouts.

Our technique for proving Theorem~\ref{thm:11} (considered in Section~\ref{sect:sub}) is quite different from
the one used by Dujmovi{\'c} and Wood~\cite{DW05} for proving the earlier (weaker) result.
We make use of the so-called concentric representation of planar graphs.
While the existence of such a representation for a planar graph is known, we extend the representation by finding
a suitable order for the vertices, and show that all planar graphs admit the extended representation.
We are not aware of any work that uses concentric representations in the context of linear layouts.

Section~\ref{sect:disc} concludes the paper with a discussion of our experiments, possible future directions, and
interesting open problems.

\paragraph{Related Work.}

Although there exists numerous works on stack and queue layouts of graphs (refer to~\cite{DW04} for a detailed list of references), 
the concept of mixed layouts received much less attention. Heath and Rosenberg~\cite{HR92} suggest to study such generalized
layouts and present Conjecture~\ref{conj:HR}, which is a topic of this paper.
Dujmovi{\'c} and Wood~\cite{DW05} investigate mixed layouts of graph subdivisions. They show that
every graph $G$ (not necessarily planar) has an $s$-stack $q$-queue subdivision with
$\Oh(\log sn(G))$ or $\Oh(\log qn(G))$ vertices per edge, where $sn(G)$ and $qn(G)$ are
the stack and queue numbers of $G$, respectively. Enomoto and Miyauchi~\cite{EM14} improve
the constants of the bounds for the numbers of division vertices per edge.

For the case of planar graphs, Dujmovi{\'c} and Wood~\cite{DW05} show that four division vertices per edge
are sufficient to construct a mixed $1$-stack $1$-queue; the bound is improved by Theorem~\ref{thm:11}.
Our result mimics the fact that every planar graph with one division vertex per edge
has a $2$-stack layout, as such graphs are bipartite~\cite{Over98}.
Also related is a work by Auer~\cite{Auer14} who study \df{dequeue} layouts of planar graphs, and
prove that a planar graph admits a dequeue layout if and only if it contains a Hamiltonian path.
Since a dequeue may be simulated by two stacks and one queue, such graphs also admit a $2$-stack $1$-queue
layout. To the best of our knowledge, it is open whether every planar graph has a $2$-stack $1$-queue layout.

\section{A Counterexample for Conjecture~1}
\label{sect:ce}

A \df{vertex ordering} of a graph $G=(V,E)$ is a total order of the vertex set $V$. In a vertex ordering
$<$ of $G$, let $L(e)$ and $R(e)$ denote the endpoints of an edge $e \in E$ such that $L(e) < R(e)$.
Consider two edges $e, f \in E$. If $L(e) < L(f) < R(e) < R(f)$ then $e$ and $f$ \df{cross}, 
and if $L(e) < L(f) < R(f) < R(e)$ then $e$ and $f$ \df{nest}. In the latter case, we also say that $e$ \df{covers} $f$.
It is convenient to express the total order $<$
by permutation of vertices $[v_1, v_2, \dots, v_{|V|}]$, where $v_1 < v_2 < \dots < v_{|V|}$.
This notion extends to a subset of vertices in
the natural way. Thus, two edges, $e$ and $f$, cross if the order is $[L(e), L(f), R(e), R(f)]$, and
they nest if the order is $[L(e), L(f), R(f), R(e)]$.
A \df{stack} (resp. \df{queue}) is a set of edges $E' \subset E$ such that no two edges in $E'$ cross (nest).
A \df{mixed} layout of a graph is a pair $(<, \{\Sh, \Qh\})$, where $<$ is a vertex ordering of $G$, 
and $\{\Sh, \Qh\}$ is a partition of $E$ into a stack $\Sh$ and a queue $\Qh$. 

Our counterexample for Conjecture~\ref{conj:HR} is depicted in Fig.~\ref{fig:mixed}. The graph, $\Gc$,
is built from $19$ copies of a gadget, $H$, by identifying two vertices, $A$ and $B$.
The graph consists of $173$ vertices and $361$ edges.
Let us introduce some definitions for the graph.
Every copy of gadget $H$ consists of two \df{twins}, $s$ and $t$, connected by a \df{twin
edge}, $(s, t)$. Each pair of twins is connected by $A$, $B$, and seven degree-2 vertices, $x_1, \dots, x_7$, that we call
\df{connectors}. The set of connectors corresponding to $s$ and $t$ is denoted by $C_{s,t}$.
Now we prove the main result of the section.

\begin{figure}[t]
    \centering
    \begin{subfigure}[b]{.4\linewidth}
        \includegraphics[page=1]{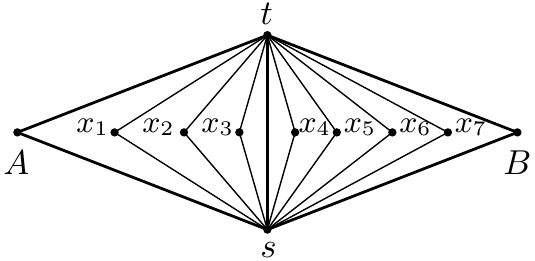}
        \caption{Gadget $H$}       
    \end{subfigure}
    \hfill
    \begin{subfigure}[b]{.4\linewidth}
        \includegraphics[page=2]{pics/mixed}
        \caption{A complete graph $\Gc$}
        \label{fig:ce}
    \end{subfigure}
    \caption{A graph that does not admit a mixed $1$-stack $1$-queue layout.}
    \label{fig:mixed}
\end{figure}

\begin{backInTimeThm}{thm-ce}
    \begin{theorem}
        There exists a planar graph that does not admit a mixed $1$-stack $1$-queue layout.
    \end{theorem}
\end{backInTimeThm}

\begin{proof}
    The proof is by contradiction; we assume that there exists a mixed layout of graph $\Gc$ shown in Fig.~\ref{fig:ce}.
    Using symmetry, we may assume that in the mixed layout of the graph $A < B$, $s_i < t_i$ for all $1 \le i \le 19$, 
    and $s_1 < s_2 < \dots < s_{19}$. Let us analyze possible relative orderings of vertices $A, B$ and 
    two twins, $s$ and $t$ in a gadget $H$. It is easy to see that there are only six permutations of the vertices:
    (i)~$[s, t, A, B]$; (ii)~$[s, A, t, B]$; (iii)~$[s, A, B, t]$; (iv)~$[A, s, t, B]$; (v)~$[A, s, B, t]$; (vi)~$[A, B, s, t]$.
    Since graph $\Gc$ contains $19$ pairs of twins, there exist at least four twin pairs that
    form the same permutation with $A$ and $B$.
    Therefore, to prove the claim of the theorem, it is sufficient to show impossibility of a mixed layout
    with four twin pairs forming the same permutation. 
    Permutations (i) and (vi) are considered in Case~\ref{cs:2}, as they are symmetric.
    Permutations (ii) and (v) are considered in Case~\ref{cs:4}. 
    Permutation (iii) is considered in Case~\ref{cs:1}. 
    Permutation (iv) is considered in Case~\ref{cs:3}. 
\end{proof}

Before moving to the case analysis, we prove three lemmas that are common for the proofs of all the cases.

\begin{lemma}
    \label{lm:RR}
    Assume that a vertex ordering of graph $\Gc$ contains $[v_1, s, t, v_2]$ with edge $(v_1, v_2) \in \Qh$ and
    twins $s, t$. Then the following holds:
    
    \begin{enumerate}[label={\bf\thelemma\alph*}]
        \item the order is $[v_1, s, t, v_2, x_1, x_2, x_3]$ or $[x_1, x_2, x_3, v_1, s, t, v_2]$
        for some connectors $x_1, x_2, x_3 \in C_{s,t}$; that is,        
        at least three of the connectors are either before $v_1$ or after $v_2$ in the order;
        \label{lm:RR1}
        
        \item $(s, x_i) \in \Sh$ and $(t, x_i) \in \Qh$, or $(s, x_i) \in \Qh$ and $(t, x_i) \in \Sh$
        for some $x_i \in C_{s,t}$, $1 \le i \le 3$; that is, at least one of the connectors is adjacent 
        to a queue edge and a stack edge. 
        \label{lm:RR2}
    \end{enumerate}
\end{lemma}

\begin{proof}
    For the first part of the lemma, assume that three of the connectors corresponding to $s$ and $t$ are
    between $v_1$ and $v_2$; that is, $v_1 < x_5 < x_6 < x_7 < v_2$ for some connectors $x_5, x_6, x_7 \in C_{s,t}$.
    Since the graph induced by vertices $s, t, x_5, x_6, x_7$ is not $1$-stack (that is, outerplanar), at
    least one of edges, $(s, x_5)$, $(s, x_6)$, $(s, x_7)$, $(t, x_5)$, $(t, x_6)$, $(t, x_7)$, is a queue edge.
    However, this edge is covered by $(v_1, v_2) \in \Qh$, a contradiction.
    
    For the second part of the lemma, assume the order is $[v_1, s, t, v_2, x_1, x_2, x_3]$ (the proof for the other
    order is symmetric). 
    Suppose that none of the connectors is adjacent to both queue and stack edges. Hence, there are two connectors, say $x_1$ and
    $x_2$, with edges assigned to a queue or to a stack. However, 
    one of $(s, x_1)$ and $(t, x_2)$ is a queue edge, as the two edges cross. Similarly, one of edges 
    $(s, x_2)$ and $(t, x_1)$ is a stack edge, as the edges are nested, a contradiction.
\end{proof}    

\begin{lemma}
\label{lm:RB}
Assume that a vertex ordering of graph $\Gc$ contains $[v_1, u_1, s, t, u_2, v_2]$ with edges $(v_1, v_2) \in \Qh$, $(u_1, u_2) \in \Sh$
and twins $s, t$. Then $\Gc$ does not admit a mixed layout.
\end{lemma}

\begin{proof}
    By Lemma~\ref{lm:RR1} applied for vertices $v_1, s, t, v_2$, there exists a connector $x \in C_{s,t}$
    such that $x < v_1$ or $x > v_2$. By Lemma~\ref{lm:RR2},
    one of edges $(s, x)$, $(t, x)$ is a stack edge. However, this edge crosses
    stack edge $(u_1, u_2)$, a contradiction; see Fig~\ref{fig:lemmaRB}.
\end{proof}

\begin{figure}[t]
    \centering
    \begin{subfigure}[b]{.4\linewidth}
        \includegraphics[page=2,width=0.99\textwidth]{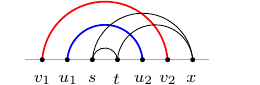}
        \caption{Lemma~\ref{lm:RB}}
        \label{fig:lemmaRB}
    \end{subfigure}
    \hfill
    \begin{subfigure}[b]{.4\linewidth}
        \includegraphics[page=3,width=0.99\textwidth]{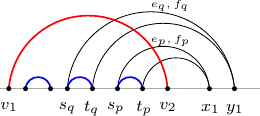}
        \caption{Lemma~\ref{lm:R}}
        \label{fig:lemmaR}
    \end{subfigure}
    \caption{Impossible configurations for a mixed layout of graph $\Gc$, as
        shown by Lemmas~\ref{lm:RB} and \ref{lm:R}. Stack edges are blue and queue edges are red.}
    \label{fig:lemmas}
\end{figure}

\begin{lemma}
    \label{lm:R}
    Assume that for three pairs of twins $s_i$, $t_i$, $1 \le i \le 3$, a vertex ordering of graph $\Gc$
    contains $[v_1, s_i, t_i, v_2]$, where edge $(v_1, v_2) \in \Qh$. Then $\Gc$ does not admit a mixed layout.
\end{lemma}

\begin{proof}
    Notice that all twin edges, $(s_i, t_i)$ for $1 \le i \le 3$, are stack edges, as they are covered by $(v_1, v_2) \in \Qh$.
    Moreover, the edges do not nest each other, as otherwise the two nested edges together with $(v_1, v_2)$ form
    a configuration as in Lemma~\ref{lm:RB}, which is impossible.
    Thus, we have three non-nested pairs of twins, that is, up to renumbering
    the order is $[v_1, s_1, t_1, s_2, t_2, s_3, t_3, v_2]$.
    
    Let us apply Lemma~\ref{lm:RR1} for the three pairs of twins and edge $(v_1, v_2) \in \Qh$. There are
    two triples of connectors, $x_1, x_2, x_3 \in C_{s_p,t_p}$ and $y_1, y_2, y_3 \in C_{s_q,t_q}$ for $p, q \in \{1,2,3\}$,
    that are all either before $v_1$ or after $v_2$ in the order. Without loss of generality, we assume
    $v_2 < x_j$ and $v_2 < y_j$ for all $1\le j \le 3$.
    
    By Lemma~\ref{lm:RR2} applied for twins $s_p$ and $t_p$, one of the connectors, say $x_1$, is adjacent to 
    a queue edge, $e_p$,  and to a stack edge, $f_p$. Similarly, a connector of $s_q$ and $t_q$, say $y_1$, is adjacent to 
    a queue edge, $e_q$, and to a stack edge, $f_q$; see Fig.~\ref{fig:lemmaR}.
    However, it is not possible to assign these four edges to $\Sh$ and $\Qh$:
    If $x_1 < y_1$, then the two queue edges, $e_p$ and $e_q$, nest; 
    If $x_1 > y_1$, then the two stack edges, $f_p$ and $f_q$, cross.
\end{proof}

Now we are ready to analyze the cases proving Theorem~\ref{thm:ce}.

\begin{cs}[$sABt$]
    \label{cs:1}
    Assume that the vertex ordering is $[s_i, A, B, t_i]$ for twins $s_i, t_i$ for all $1 \le i \le 4$.
    Then graph $\Gc$ does not admit a mixed layout.
\end{cs}    

\begin{proof}
    Let us assume that the vertex ordering is $[s_1, s_2, s_3, s_4, A, B, d_1, d_2, d_3, d_4]$, where 
    $d_i \in \{t_1, t_2, t_3, t_4\}$ for all $1 \le i \le 4$.
    Note that $\Gc$ contains edges $(A, d_i)$ and $(B, d_i)$ for all $1 \le i \le 4$.
    
    Start with a pair of edges $(A, d_4)$ and $(B, s_1)$; since they cross, one of the edges is a queue edge. Without
    loss of generality, we may assume
    that $(A, d_4) \in \Qh$. Hence, all edges covered by $(A, d_4)$ are stack edges; that is, 
    $(B, d_3)$, $(B, d_2)$, $(B, d_1) \in \Sh$. It follows that all edges crossing the three edges are in 
    the queue: $(A, d_2)$, $(A, d_1) \in \Qh$; see Fig.~\ref{fig:case1}.
    
    Now let $s_x$ be a twin of $d_2$. Edge $(s_x, d_2)$ is a stack edge since it covers $(A, d_1) \in \Qh$. However,
    $(s_x, d_2)$ crosses $(B, d_3) \in \Sh$, a contradiction.
\end{proof}    

\begin{figure}[t]
    \centering
    \begin{subfigure}[b]{.4\linewidth}
        \includegraphics[page=1,width=0.99\textwidth]{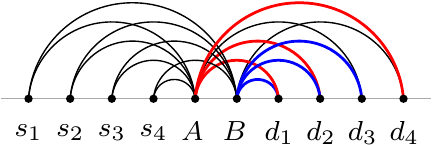}
        \caption{Case~\ref{cs:1}}
        \label{fig:case1}
    \end{subfigure}
    \hfill
    \begin{subfigure}[b]{.4\linewidth}
        \includegraphics[page=4,width=0.99\textwidth]{pics/cases}
        \caption{Case~\ref{cs:4}}
        \label{fig:case4}
    \end{subfigure}
    \caption{An illustration for the proofs of cases of Theorem~\ref{thm:ce}. 
        Stack edges are blue and queue edges are red.}
\end{figure}

\begin{cs}[$AsBt$]
    \label{cs:4}
    Assume that the vertex ordering is $[A, s_i, B, t_i]$ for twins $s_i, t_i$ for all $1 \le i \le 4$.
    Then graph $\Gc$ does not admit a mixed layout.
\end{cs}

\begin{proof}
    Let us assume that the vertex ordering is $[A, s_1, s_2, s_3, s_4, B, d_1, d_2, d_3, d_4]$, where 
    $d_i \in \{t_1, t_2, t_3, t_4\}$ for all $1 \le i \le 4$.

    Suppose that $e=(A, d_2)$ is a stack edge. Then $(B, d_3) \in \Qh$, as it crosses $e$.
    This is impossible, as twin edge $(s_x, d_4)$ (for some $1\le x \le 4$) crosses a stack edge, $(A, d_2)$, 
    and covers a queue edge, $(B, d_3)$. Hence, $(A, d_2)$ is a queue edge.
    
    Since $(A, d_2) \in \Qh$, all nested edges are in the stack; in particular, $(s_1, B) \in \Sh$ and $(s_y, d_1) \in \Sh$, where
    $s_y$ is the twin of $d_1$. Notice that in order for edges $(s_1, B)$ and $(s_y, d_1)$ to be non-crossing, 
    $s_y$ should be equal to $s_1$; see Fig.~\ref{fig:case4}. It follows that edge $(B, d_2)$ is a queue edge since
    it crosses $(s_1, d_1)$.
    Finally, we observe that the twin edge of $d_3$, $(s_z, d_3)$ for some $2 \le z \le 4$,
    crosses a stack edge, $(s_1, d_1)$, and covers a queue edge, $(B, d_2)$, which is not possible.
\end{proof}    

\begin{cs}[$AstB$]
    \label{cs:3}
    Assume that the vertex ordering is $[A, s_i, t_i, B]$ for twins $s_i, t_i$ for all $1 \le i \le 4$.
    Then graph $\Gc$ does not admit a mixed layout.
\end{cs}

\begin{proof}
    Let us assume that the vertex ordering is $[A, d_1, d_2, \dots, d_8, B]$,  
    where $s_i, t_i \in \{d_1, \dots, d_8\}$ for all $1 \le i \le 4$.
    
    Since edges $(A, d_8)$ and $(d_1, B)$ cross, one of them is a queue edge. Without
    loss of generality, we may assume $(A, d_8) \in \Qh$.
    Consider seven vertices $d_1, d_2, \dots, d_7$. It is easy to see that they form three
    pairs of twins (while the forth pair is formed with $d_8$). By Lemma~\ref{lm:R} applied for
    $(A, d_8) \in \Qh$ and the twins, it is impossible.
\end{proof}    

\begin{cs}[$ABst$]
    \label{cs:2}
    Assume that the vertex ordering is $[A, B, s_i, t_i]$ for twins $s_i, t_i$ for all $1 \le i \le 4$.
    Then graph $\Gc$ does not admit a mixed layout.
\end{cs}    

\begin{proof}
    Let us assume that the vertex ordering is $[A, B, d_1, d_2, \dots, d_8]$,  
    where $s_i, t_i \in \{d_1, \dots, d_8\}$ for all $1 \le i \le 4$.

    Suppose that edge $e=(A, d_7) \in \Qh$; then edge $(B, d_6) \in \Sh$, as it is covered by $e$. Additionally,
    we have five vertices, $d_1, d_2, d_3, d_4, d_5$, which form at least one pair of twins.
    This pair of twins together with $(A, d_7) \in \Qh$ and $(B, d_6) \in \Sh$ form a configuration as in Lemma~\ref{lm:RB},
    which is impossible. Therefore, edge $(A, d_7) \in \Sh$ and the crossing edge, $(B, d_8)$, is a queue edge.
    
    Consider vertices $d_1, \dots, d_7$. There are three pairs of twins formed by the vertices; all of the
    pairs are covered by $(B, d_8) \in \Qh$, contradicting Lemma~\ref{lm:R}.
\end{proof}    

\section{Mixed Layouts of Planar Subdivisions}
\label{sect:sub}

In this section we prove Theorem~\ref{thm:11}. To this end,
we utilize a special representation of a planar graph, which 
is called \df{ordered concentric}.
In such a representation, the vertices of a graph are 
laid out on a set of circles around a specified origin vertex, so that each circle contains exactly
the vertices with the same graph-theoretic distance to the origin; see Fig.~\ref{fig:cr}. To construct such a
representation, we begin with an arbitrary 
vertex of the graph as the origin, and consider a planar embedding of the graph with the origin on the outer face.
The layers of vertices are formed by a breadth-first search starting at the origin, and
the edges are routed without crossings and connect vertices of the same layer or vertices
of two consecutive layers. Formally an ordered concentric representation is defined next.
We assume that for a graph $G=(V, E)$, $dist_G(u, v)$ is the graph-theoretic
distance between vertices $u, v \in V$.

\begin{figure}[t]
    \centering
    \begin{subfigure}[b]{.48\linewidth}
        \includegraphics[page=1,width=\textwidth]{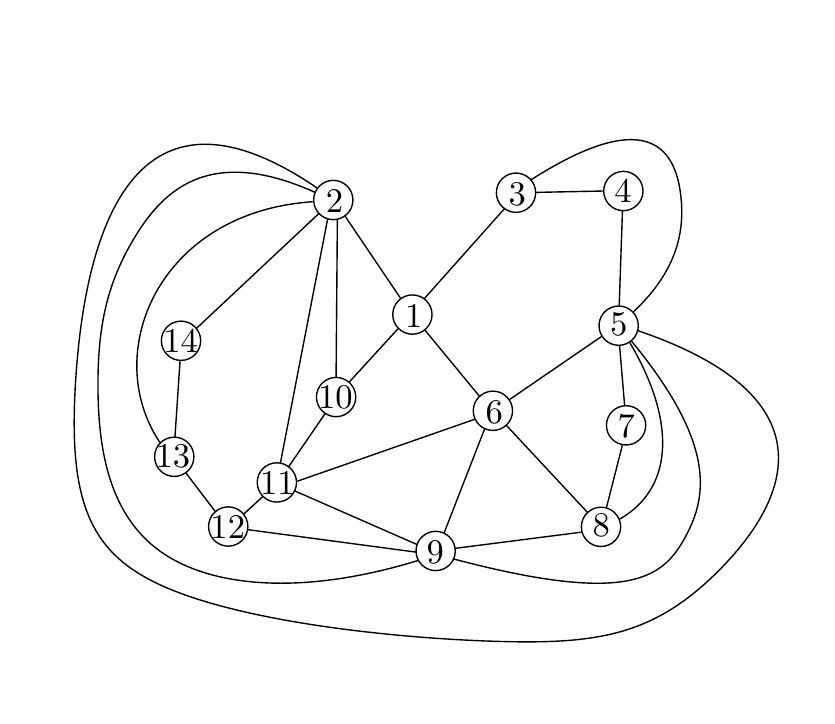}
        \caption{}
        \label{fig:cr1}
    \end{subfigure}
\hfill
    \begin{subfigure}[b]{.48\linewidth}
        \includegraphics[page=2,width=\textwidth]{pics/levels}
        \caption{}
        \label{fig:cr2}
    \end{subfigure}
    \caption{(a)~A plane graph and (b) its ordered concentric representation with the origin $v^*=1$.
    Concentric circles are shown dashed.}
    \label{fig:cr}
\end{figure}

\begin{definition}
    \label{def:con}
    Let $G=(V, E)$ be a connected planar graph with a specified vertex $v^* \in V$. An \dfb{ordered concentric
    representation of $G$ with the origin $v^*$}, denoted by $\Gamma^o$, is a drawing of $G$ with the following properties:

\begin{enumerate}[label={(\roman*)}]
    \item The drawing is planar with vertex $v^*$ lying on the outer face;
    
    \item $V_i = \{x \in V~|~dist_{G}(v^*, x) = i\}$ for all $0 \le i \le k$ and some $k \in \mathbb{N}$.
    The set $V_i$ is called the \dfb{\emph{i}-th level} of $G$. For every level $V_i$, $1 \le i \le k$, 
    the vertices of $V_i$ are arranged in
    a sequence $x_1, x_2, \dots ,x_r$ with $|V_i| = r$. 
    
    In the drawing, the vertices are 
    laid out on a closed curve in the order; the curve is called the \dfb{\emph{i}-th circle}.
    The vertices of $V_j$, $j < i$ are located inside the area bounded by the $i$-th circle,
    while the vertices of $V_j$, $j > i$ are located outside the area bounded by the $i$-th circle.

    \item For every edge $(u, w) \in E$ with $u \in V_i$ and $w \in V_j$, it holds that either $i=j$,
    in which case the edge is called a \dfb{level} edge, or $|i-j|=1$, corresponding to a \dfb{non-level} edge.
    
    Every level edge $(u, w) \in E$, $u, w \in V_i$ is realized as a curve routed outside the $i$-th circle.
    Every non-level edge $(u, w) \in E$, $u \in V_i$, $w \in V_{i+1}$ is realized as a curve consisting of
    at most two pieces: the first (required) piece is routed between the $i$-th and the $(i+1)$-th circles, and
    the second (optional) piece is routed outside the $(i+1)$-th circle.
\end{enumerate}    
\end{definition}    

Notice that the notion of concentric representations is related to radial drawings~\cite{BBF05}. 
The main difference is that in radial drawings ``monotonicity'' of edges is required;
equivalently, every edge shares at most one point with a circle in a radial drawing. In contrast, some edges of
a concentric representation may cross a circle multiple times; for example, see an
edge $(5, 8)$ in Fig.~\ref{fig:cr2}. As a result, a radial drawing may not exist for a planar graph, while
an ordered concentric representation can always be constructed as shown by Lemma~\ref{lm:con}.

Another closely related concept is a (non-ordered) concentric representation of a planar graph~\cite{Ull84,Hro13}.
Such a representation is defined similarly, except that the vertices of each level form a cyclic sequence and
the origin vertex is not required to lie on the outer face. In a sense, Definition~\ref{def:con} provides a
refinement of a concentric representation, as it dictates the (non-cyclic) order of the vertices of every level.
The next lemma shows that every planar graph admits an ordered concentric representation.

\begin{lemma}
    \label{lm:con}    
    For every connected planar graph $G=(V, E)$ and every $v^* \in V$, there exists an ordered concentric representation of $G$ 
    with the origin $v^*$.
\end{lemma}    

\begin{proof}
    We start by constructing a breadth-first search tree, $T$, of $G$ rooted at $v^*$.    
    Consider an arbitrary combinatorial embedding of $G$ (that is, cyclic orders of edges around each vertex), 
    and draw $T$ on a set of horizontal lines respecting the planar embedding. A line with $y$-coordinate $=i$
    contains vertices $V_i$, $0 \le i \le k$ with $k = \max_{u \in V} dist_G(v^*, u)$. The order
    of the vertices along each line is defined by the embedding of $G$.
    Notice that the drawing is planar and satisfies Definition~\ref{def:con}. Here the circles are formed by 
    connecting the leftmost and the rightmost vertices of each horizontal line with a curve surrounding the drawing; see Fig.~\ref{fig:con1}. All the edges of $T$ are drawn as straight-line segments between consecutive
    levels, that is, they are non-level edges.
    Next we show how to draw the remaining edges of $E$ while preserving the properties of Definition~\ref{def:con}.
    
    To this end, we maintain the following invariant: For every face $f$ of the currently drawn graph, $H$,
    there exists a vertical line segment of length $\eps > 0$ such that every vertex $u \in f$ can be connected to
    every point of the segment via a curve, which is monotone in the $y$-direction, while avoiding crossing with the edges of $H$.
    Since the drawing of $T$ defines only one face, it is clear that the segment with endpoints $(x, k)$ and $(x, k+1)$
    (for an arbitrary $x \in \mathbb{R}$) satisfies the invariant; see Fig.~\ref{fig:con1}. Let us show how to draw the next
    edge. Assume that an edge, $(u, w) \in E$, belongs to a face $f$ of $H$. Due to the invariant, there exists a
    line segment with endpoints $p_0=(x, y_0)$ and $p_1=(x, y_1)$ (for some $x, y_0, y_1 \in \mathbb{R}$) that is reachable
    from both $u$ and $w$. We identify point $p=(x, (y_0 + y_1) / 2)$ on the segment and route the edge
    along the curves connecting $u$ to $p$ and then $p$ to $w$. The edge splits $f$ into two faces
    such that the condition of the invariant can be satisfied using segments $p_0, p$ and $p, p_1$; see Fig.~\ref{fig:con2}.    
    To complete the proof, we
    observe that the edge also satisfies Definition~\ref{def:con}. If $dist_G(v^*, u) = dist_G(v^*, w)$, then $(u, w)$
    is a level edge. 
    Otherwise, if $|dist_G(v^*, u) - dist_G(v^*, w)| = 1$, then $(u, w)$ is a non-level edge
    represented by a two-piece curve.
\end{proof}    

\begin{figure}[t]
    \centering
    \begin{subfigure}[b]{.48\linewidth}
        \includegraphics[page=1,width=\textwidth]{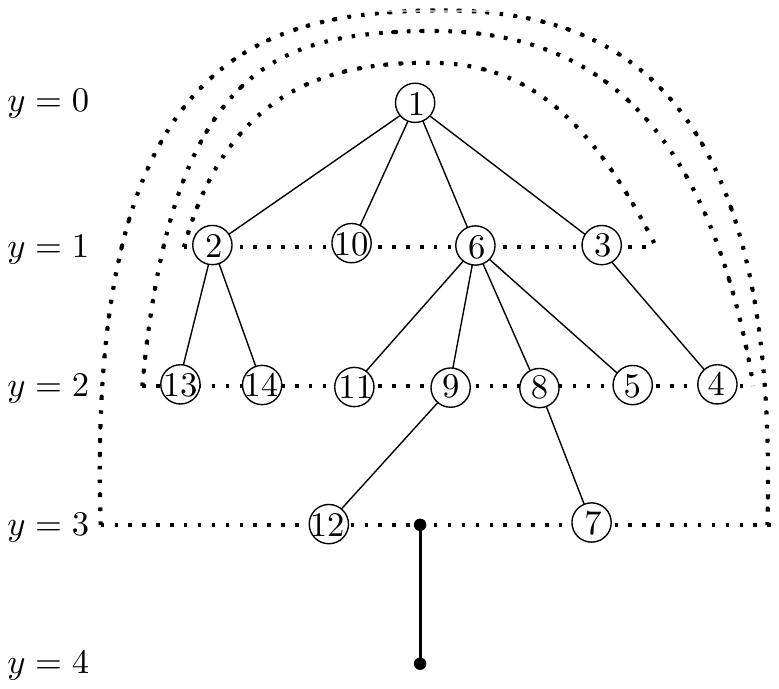}
        \caption{}
        \label{fig:con1}
    \end{subfigure}
\hfill
    \begin{subfigure}[b]{.48\linewidth}
        \includegraphics[page=2,width=\textwidth]{pics/levels2}
        \caption{}
        \label{fig:con2}
    \end{subfigure}
    \caption{(a)~A starting point for construction an ordered concentric representation for graph shown
        in Fig.~\ref{fig:cr1}, as described in Lemma~\ref{lm:con}.
        (b)~Maintaining the invariant of Lemma~\ref{lm:con} while drawing edge $(2,9)$ and then edge $(5,8)$.}
    \label{fig:con}
\end{figure}

Now we are ready to prove the main result of the section. Our construction
of a mixed layout for a given graph $G$ is as follows. We start with an ordered concentric representation, $\Gamma^o$, of $G$, which
is created by a breadth-first search starting at an arbitrary vertex $v^*$. We distinguish three types of edges in~$\Gamma^o$:
(a)~level edges with both endpoints belonging to the same level of $\Gamma^o$, 
(b)~\df{short} non-level edges whose curves are routed between consecutive levels of~$\Gamma^o$, and
(c)~\df{long} non-level edges whose curves cross some circles of $\Gamma^o$.
Our goal is to keep the edges of type (a) in the stack and the
edges of type (b) in the queue. The edges of type
(c) shall be subdivided into a stack edge and a queue edge. The order of vertices in the mixed
layout is constructed from the levels of $\Gamma^o$: we place the origin $v^*$, which is followed by 
the vertices of $V_1$ (in the order given by the ordered concentric representation), followed by the
vertices of $V_2$ etc. The correctness of the mixed layout follows from the planarity of $\Gamma^o$; see Fig.~\ref{fig:pr}.
Next we provide a formal proof.

\begin{backInTimeThm}{thm-11}
    \begin{theorem}
        Every planar graph admits a mixed $1$-stack $1$-queue subdivision with one division vertex per edge.
    \end{theorem}
\end{backInTimeThm}

\begin{proof}
    Let $G = (V, E)$ be a planar graph and $v^* \in V$. We assume that $G$ is connected; otherwise, 
    each connected component of $G$ can be processed individually.
    Using Lemma~\ref{lm:con}, construct an ordered 
    concentric representation $\Gamma^o$ of 
    the graph with the origin $v^*$ and a set of levels $V_i$, $0 \le i \le k$ for some 
    $k \ge 0$.
    
    Consider an edge $(u, w) \in E$ with $u \in V_i$ and $w \in V_j$. Since the levels are constructed
    with a breadth-first search, it holds that $i=j$ or $|i-j|=1$. 
    Let $P_{u,w}$ be the (ordered) sequence of levels 
    of $\Gamma^o$ such that the corresponding circles share a point with the curve realizing edge $(u, w)$. 
    By Definition~\ref{def:con}, $E = E_a \cup E_b \cup E_c$, where
    \begin{itemize}
        \item $E_a = \{(u, w) : u, w \in V_i \text { for some } 1 \le i \le k$ with $P_{u,w}=(V_i, \dots, V_i)$ in $\Gamma^o \}$;
        \item $E_b = \{(u, w) : u \in V_i, w \in V_{i+1}$ with $P_{u,w}=(V_i, V_{i+1})$ in $\Gamma^o \}$;
        \item $E_c = \{(u, w) : u \in V_i, w \in V_{i+1}$ with $P_{u,w}=(V_i, V_{y_1}, \dots, V_{y_d}, V_{i+1})$ in $\Gamma^o$ 
        for some $d \ge 1$, $y_1 = i+1$ and $y_t > i+1$ for $2 \le t \le d \}$.
    \end{itemize}

    We construct a subdivision $G^s=(V^s, E^s)$ as follows. For an edge $e=(u, w) \in E_c$, 
    consider the crossing point, $y_1(e)$, between the $(i+1)$-th circle and the curve realizing $e$ in $\Gamma^o$
    (which corresponds to the second element, $V_{y_1}$, of $P_{u, w}$).
    Let $$V^s = V \cup \bigcup_{e \in E_c} \{y_1(e)\} \text{ and } E^s = E_a \cup E_b \cup E_{c_1} \cup E_{c_2}, \text{ where }$$
    $$E_{c_1} = \bigcup_{e \in E_c} \{(u, y_1(e))\} \text{ and }
    E_{c_2} = \bigcup_{e \in E_c} \{(y_1(e), w)\}.$$
    
    
    In order to construct a mixed layout of $G^s$, we use the following order:
    $$
        \sigma = (v^*,~~ x_1^1, x_2^1, \dots, x_{r_1}^1,~~ x_1^2, x_2^2, \dots, x_{r_2}^2,~~ x_1^3, x_2^3, \dots, x_{r_3}^3, \dots),
    $$
    where $\{x_1^i, x_2^i, \dots, x_{r_i}^i\} = V_i^s \supseteq V_i$ are the vertices of the $i$-th level of $G^s$ 
    in the order given by $\Gamma^o$.
    All edges of $E_a$ and $E_{c_2}$ are stack edges; that is, $e \in \Sh$ for every $e \in E_a \cup E_{c_2}$.
    All the remaining edges are queue edges; that is, $e \in \Qh$ for every $e \in E_b \cup E_{c_1}$. Next we prove
    the correctness of the construction.
    
    Let us show that all stack edges are crossing-free with respect to the specified order, $\sigma$.
    Assume two edges, $(u_1, w_1), (u_2, w_2) \in \Sh$, cross each other, that is
    $u_1 < u_2 < w_1 < w_2$ with respect to $\sigma$.
    Observe that all edges in the stack are the edges of the same level in the ordered concentric representation.
    Thus, $u_1, w_1 \in V_i^s$ and $u_2, w_2 \in V_j^s$ for some $0 \le i, j \le k$. However, the levels
    are arranged consecutively in $\sigma$, which means that $i=j$ and two edges of the same level cross. This is impossible, as
    all vertices of the same level of $\Gamma^o$ and the corresponding level edges form an outerplanar graph by Definition~\ref{def:con}.
    
    Finally, let us show that all queue edges are non-nested with respect to $\sigma$. 
    Assume that two edges, $(u_1, w_1), (u_2, w_2) \in \Qh$, nest each other so that
    $u_1 < u_2 < w_2 < w_1$. Since the queue edges belong to consecutive levels in the ordered concentric representation,
    it holds that $u_1 \in V_i^s$, $w_1 \in V_{i+1}^s$, $u_2 \in V_j^s$, $w_2 \in V_{j+1}^s$ for some
    $0 \le i, j \le k$. Since the levels do not overlap in $\sigma$, it holds that $i = j$.
    Hence, the two edges are routed between the same consecutive levels, $V_i^s$ and $V_{i+1}^s$, and therefore, cross
    each other, which violates the planarity of $\Gamma^o$.
\end{proof}    

\begin{figure}[t]
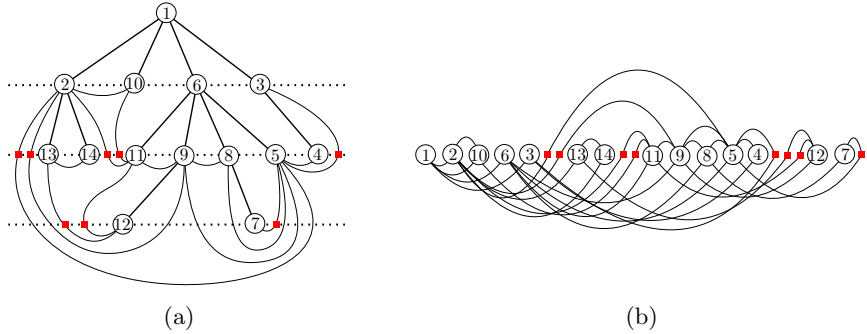

    \centering
    \begin{subfigure}[b]{.495\linewidth}
        \includegraphics[page=3,width=\textwidth]{pics/levels2}
        \caption{}
        \label{fig:pr1}
    \end{subfigure}
    \begin{subfigure}[b]{.495\linewidth}
        \includegraphics[page=4,width=\textwidth]{pics/levels2}
        \caption{}
        \label{fig:pr2}
    \end{subfigure}
    \caption{(a)~An ordered concentric representation for graph shown
        in Fig.~\ref{fig:cr1}. Red squares are subdivision vertices introduced for long non-level edges.
        (b)~A~mixed $1$-stack $1$-queue layout of the graph constructed as described in Theorem~\ref{thm:11}.}
    \label{fig:pr}
\end{figure}

\section{Discussion}
\label{sect:disc}

In this paper we resolved a conjecture by Heath and Rosenberg~\cite{HR92} by providing a graph
that does not admit a mixed $1$-stack $1$-queue layout. The graph contains $173$ vertices, and
a reasonable question is what is the size of the smallest counterexample. In an attempt to answer
the question, we implemented an exhaustive search 
algorithm\footnote{An online tool and the source code for testing linear embeddability of graphs is available at \url{http://be.cs.arizona.edu}}
(based on the SAT formulation of
the linear embedding problem suggested by Bekos, Kaufmann, and Zielke~\cite{BKZ15}) and
tested {\it all} $977,526,957$ maximal planar graphs with $|V| \le 18$. It turns out that
all such graphs have a mixed $1$-stack $1$-queue layout. The evaluation suggests that mixed
layouts are more ``powerful'' than pure stack and queue layouts, as there exist fairly small
graphs that do not admit $2$-stack and $2$-queue layouts.
The smallest planar graph requiring three stacks contains $11$ vertices, and the smallest planar graph
requiring three queues contains $14$ vertices; 
see Figs.~\ref{fig:stack3}, \ref{fig:queue3} and Appendix~\ref{sec:app}.
We were able to find a smaller counterexample for Conjecture~\ref{conj:HR}; see Fig.~\ref{fig:mixed3}. This instance consists
of $|V|=37$ vertices and $|E|=77$ edges, and has a similar structure as the graph in Theorem~\ref{thm:ce}.
However, showing that the graph does not admit a mixed layout requires significantly more effort.

\begin{figure}[t]
    \centering
    \begin{subfigure}[b]{.23\linewidth}
        \includegraphics[page=1,width=0.9\textwidth]{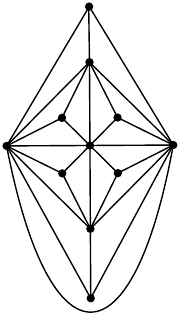}
        \caption{}
        \label{fig:stack3}
    \end{subfigure}
    \hfill
    \begin{subfigure}[b]{.32\linewidth}
        \includegraphics[page=2,width=0.9\textwidth]{pics/smallest_ce}
        \caption{}
        \label{fig:queue3}
    \end{subfigure}
    \hfill
    \begin{subfigure}[b]{.34\linewidth}
        \includegraphics[page=3,width=0.9\textwidth]{pics/smallest_ce}
        \caption{}
        \label{fig:mixed3}
    \end{subfigure}
    \caption{The smallest planar graphs that require (a)~3 stacks, (b)~3 queues. 
        (c)~A~graph with $|V|=37$ and $|E|=77$ that does not admit a mixed layout.}
\end{figure}

An interesting future direction is to consider {\it bipartite} planar graphs. We noticed that
all our counterexamples contain triangles, which seem to be important for non-embeddability.
Based on our experiments, we conjecture that every bipartite planar graph admits a mixed $1$-stack $1$-queue layout.
Such a result would strengthen Theorem~\ref{thm:11}, as a subdivision of a graph with one 
vertex per edge is clearly bipartite. Also observe that a pure $2$-stack layout exists for every 
bipartite planar graph, as shown by Overbay~\cite{Over98}.

\begin{conj}
    Every bipartite planar graph admits a mixed $1$-stack $1$-queue layout.
\end{conj}

Another future direction is to study mixed $s$-stack $q$-queue layouts of planar graphs. What 
are possible values of $s$ and $q$ such that there exists a mixed $s$-stack $q$-queue layout for every
planar graph? By the result of Yannakakis~\cite{Yan89}, we know that $s=4, q=0$ is a valid option, while
Theorem~\ref{thm:ce} shows that $s=1, q=1$ is not sufficient. Here it is worth mentioning a result
by Auer~\cite{Auer14} who shows that every planar graph with a Hamiltonian path admits a mixed layout 
with $s=2$ and $q=1$. However it is open whether there exists some $s > 0$ and $q > 0$ with $2 < s + q \le 4$ realizing
all planar graphs.

\bibliographystyle{splncs03}
\bibliography{refs}

\begin{thebibliography}{10}
\providecommand{\url}[1]{\texttt{#1}}
\providecommand{\urlprefix}{URL }

\bibitem{Auer14}
Auer, C.: Planar graphs and their duals on cylinder surfaces. Ph.D. thesis,
  Universit{\"a}t Passau (2014)

\bibitem{BBF05}
Bachmaier, C., Brandenburg, F.J., Forster, M.: Radial level planarity testing
  and embedding in linear time. Journal of Graph Algorithms and Applications
  9(1),  53--97 (2005)

\bibitem{BKZ15}
Bekos, M.A., Kaufmann, M., Zielke, C.: The book embedding problem from a
  sat-solving perspective. In: Di~Giacomo, E., Lubiw, A. (eds.) International
  Symposium on Graph Drawing and Network Visualization. pp. 125--138. Springer
  (2015)

\bibitem{BK79}
Bernhart, F., Kainen, P.C.: The book thickness of a graph. Journal of
  Combinatorial Theory, Series B  27(3),  320--331 (1979)

\bibitem{BFP13}
Di~Battista, G., Frati, F., Pach, J.: On the queue number of planar graphs.
  SIAM Journal on Computing  42(6),  2243--2285 (2013)

\bibitem{Duj15}
Dujmovi{\'c}, V.: Graph layouts via layered separators. Journal of
  Combinatorial Theory, Series B  110,  79--89 (2015)

\bibitem{DW04}
Dujmovi{\'c}, V., Wood, D.R.: On linear layouts of graphs. Discrete Mathematics
  \& Theoretical Computer Science  6(2),  339--358 (2004)

\bibitem{DW05}
Dujmovi{\'c}, V., Wood, D.R.: Stacks, queues and tracks: Layouts of graph
  subdivisions. Discrete Mathematics and Theoretical Computer Science  7,
  155--202 (2005)

\bibitem{EM14}
Enomoto, H., Miyauchi, M.: Stack-queue mixed layouts of graph subdivisions. In:
  Forum on Information Technology. pp. 47--56 (2014)

\bibitem{HLR92}
Heath, L.S., Leighton, F.T., Rosenberg, A.L.: Comparing queues and stacks as
  machines for laying out graphs. SIAM Journal on Discrete Mathematics  5(3),
  398--412 (1992)

\bibitem{HR92}
Heath, L.S., Rosenberg, A.L.: Laying out graphs using queues. SIAM Journal on
  Computing  21(5),  927--958 (1992)

\bibitem{Hro13}
Hromkovi{\v{c}}, J.: Communication complexity and parallel computing. Springer
  Science \& Business Media (2013)

\bibitem{Oll73}
Ollmann, L.T.: On the book thicknesses of various graphs. In: Southeastern
  Conference on Combinatorics, Graph Theory and Computing. vol.~8, p. 459
  (1973)

\bibitem{Over98}
Overbay, S.B.: Generalized book embeddings. Ph.D. thesis, Colorado State
  University (1998)

\bibitem{RM95}
Rengarajan, S., Madhavan, C.V.: Stack and queue number of 2-trees. In:
  International Computing and Combinatorics Conference. pp. 203--212. Springer
  (1995)

\bibitem{Ull84}
Ullman, J.D.: Computational aspects of {VLSI}. Computer Science Press (1984)

\bibitem{Wie17}
Wiechert, V.: On the queue-number of graphs with bounded tree-width. Electr. J.
  Comb.  24(1),  P1.65 (2017)

\bibitem{Yan89}
Yannakakis, M.: Embedding planar graphs in four pages. Journal of Computer and
  System Sciences  38(1),  36--67 (1989)

\end{thebibliography}

\newpage
\appendix
\chapter*{\appendixname}
\section{Enumeration of Maximal Planar Graphs}
\label{sec:app}

We implemented an algorithm for constructing \emph{optimal} linear embeddings.
It is based on the SAT formulation of the problem and is capable to compute
optimal linear embeddings of graphs with hundreds of vertices within several minutes.
The source code and an online tool is available at \url{http://be.cs.arizona.edu}. 

Here we count the maximal planar graphs
having particular stack and queue numbers for small values of $n = |V|$.
It turns out that all graphs with $n \le 18$ admit a mixed $1$-stack $1$-queue layout, 
a $3$-stack layout, and a $3$-queue layout.

\newcolumntype{P}[1]{>{\raggedleft\arraybackslash}p{#1}}

\begin{table}[h]
    \centering
    \caption{The number of maximal planar graphs with a given stack and queue numbers for various sizes of the graph}
    \label{table:enumMax}
    \medskip
    
    \begin{tabular}{|P{0.5cm}|P{1.8cm} || P{1.8cm} || P{1.8cm}|P{1.6cm}|| P{1.8cm}| P{1.6cm} ||}
        \toprule
        $n$ & total & mixed & 2-stack & 3-stack & 2-queue & 3-queue \\
        \midrule
        4	& 1	& 1	& 1	& 0 & 1 & 0 \\
        5	& 1	& 1	& 1	& 0 & 1 & 0 \\
        6	& 2	& 2	& 2	& 0 & 2 & 0 \\
        7	& 5	& 5	& 5	& 0 & 5 & 0 \\
        8	& 14	& 14	& 14	& 0 & 14 & 0 \\
        9	& 50	& 50	& 50	& 0 & 50 & 0 \\
        10	& 233	& 233	& 233	& 0 & 233 & 0 \\
        11	& 1249	& 1249	& 1248	& 1 & 1249 & 0 \\
        12	& 7595	& 7595	& 7593	& 2 & 7595 & 0 \\     
        13	& 49566	& 49566	& 49536	& 30 & 49566 & 0 \\
        14	& 339722	& 339722	& 339483	& 239 & 339712 & 10 \\
        15	& 2406841	& 2406841	& 2404472	& 2369 & 2405167 & 1674 \\
        16	& 17490241	& 17490241	& 17468202	& 22039 & 17412878 & 77363 \\
        17	& 129664753	& 129664753	& 129459090	& 205663 & 127855172 & 1809581 \\
        18	& 977526957	& 977526957	& 975647292	& 1879665 & 947394711 & 30132246 \\
        \bottomrule        
    \end{tabular}
\end{table}

\end{document}